\documentclass[a4paper,twoside]{article}

\usepackage{epsfig}
\usepackage{subfigure}
\usepackage{calc}
\usepackage{amstext}
\usepackage{multicol}
\usepackage{pslatex}
\usepackage{apalike}
\usepackage{longtable,booktabs}
\usepackage{multirow}

\usepackage{soul}
\usepackage{color}

\usepackage{amssymb,amsmath,amsthm}
\newtheorem{theorem}{Theorem}
\newtheorem{corollary}{Corollary}
\newtheorem{definition}{Definition}

\usepackage{SCITEPRESS}     

\begin{document}

\title{Privacy and Fairness in Recommender Systems via Adversarial Training of User Representations}

\author{\authorname{Yehezkel S. Resheff \sup{1}, Yanai Elazar \sup{2,3}, Moni Shahar \sup{1}, and Oren Sar Shalom \sup{3}}
\affiliation{\sup{1}Intuit Tech Futures, Israel}
\affiliation{\sup{2}Bar Ilan University, Israel}
\affiliation{\sup{3}Intuit, Israel}
\email{\{hezi.resheff, yanaiela, monishahar, oren.sarshalom\}@gmail.com}
}

\keywords{Privacy, Representations, Information Leakage}

\abstract{Latent factor models for recommender systems represent users and items as low dimensional vectors. Privacy risks of such systems have previously been studied mostly in the context of recovery of personal information in the form of usage records from the training data. However, the user representations themselves may be used together with external data to recover private user information such as gender and age. In this paper we show that user vectors calculated by a common recommender system can be exploited in this way. We propose the privacy-adversarial framework to eliminate such leakage of private information, and study the trade-off between recommender performance and leakage both theoretically and empirically using a benchmark dataset. An advantage of the proposed method is that it also helps guarantee fairness of results, since all implicit knowledge of a set of attributes is scrubbed from the representations used by the model, and thus can't enter into the decision making. We discuss further applications of this method towards the generation of deeper and more insightful recommendations. 
}


\onecolumn \maketitle \normalsize \vfill

\section{\uppercase{Introduction}}
With the increasing popularity of digital content consumption, recommender systems have become a major influencer of online behavior, from what news we read and what movies we watch, to what products we buy. Recommender systems have revolutionized the way items are chosen across multiple domains, moving users from the previous active search scenario to the more passive selection of presented content as we know it today. 

A recommender system needs to fulfill two objectives in order to be able to supply relevant recommendations that will be helpful to users. Namely, it has to accurately model both the users and the items. An implication of the first objective is that recommenders aim at revealing users' preferences, their desires and wills, and even learn \emph{when} to suggest different items \cite{adomavicius2015context}, and how many times to recommend a given item. In order to be effective, a system must have an extensive view of a user, embodied as a representation that includes a substantial amount of private information (which may be either implicitly or explicitly encoded).

Indeed, in recent years we have seen a plethora of advanced methods applied to improve the personalized recommendation problem. Furthermore, modern Collaborative Filtering approaches are becoming increasingly more complex not only in algorithmic terms but also in their ability to process and use additional data. Arguably, demographic data is the most valuable source of information for user modeling. As such, it was used from the early days of recommendation systems \cite{pazzani1999framework}. New deep learning based state of the art methods also utilize demographic information \cite{covington2016deep,zhao2014we} in order to generate better and more relevant predictions. 

While this in-depth modeling of users holds great value, it may also pose severe privacy threats. One major privacy concern, especially when private information is explicitly used during training, is recovery of records from the training data by an attacker. This aspect of privacy has previously been studied in the context of recommender systems using the framework of differential privacy (see Section \ref{sec:related} below). A related but different threat, which to the best of our knowledge has not yet been addressed in this context, is the recovery of private information that was \textit{not} explicitly present in the training data. In the \textit{implicit private information} attack, either the recommendations or the user representations within a recommender system are used together with an external information to uncover private information about the users. 

For example, an attacker with access to the gender of some users could use their representations and a supervised learning approach to infer the gender of all users. In an extreme case, it would suffice for an attacker to glance at a random individual's computer screen, see what content is being suggested to them by the recommendation system, and immediately infer many demographic, financial, and other information pertaining to the individual.

In this paper we introduce the threat of \textit{implicit private information} leakage in recommender systems. We show the existence of the leakage both theoretically, and experimentally using a standard recommender and benchmark dataset. 
Finally, we propose the privacy-adversarial method of constructing a recommender from which the target private information cannot be read out, using the adversarial training method \cite{ganin2016domain}. We show the ability of this method to conceal private information and the trade-off between recommender performance and private information leakage. 

\noindent The contributions of this paper are as follows:

\begin{itemize}
  \item We describe and formalize the threat of implicit information leakage in machine learning models in general, and recommenders in particular. 
  \item We suggest the privacy-adversarial method (an adaptation of adversarial training to our setting) to eliminate the leakage, and validate the method in an extensive set of experiments. 
  \item Finally, we discuss the tangential issue of fairness in machine learning, and suggest that the privacy-adversarial method we use in this paper for the sake of privacy is capable of fixing some fairness problems previously discussed in the context of general machine learning models. 
\end{itemize}

\section{\uppercase{Related Work}}
\label{sec:related}
Since the release of the Netflix Prize dataset \cite{bennett2007netflix}, a large body of work proved that raw historical usage of users might reveal private information about them (e.g., \cite{weinsberg2012blurme,narayanan2008robust}). That is, given historical usage of an anonymized user, one can infer the user's demographics or even their identity. However, we argue that even the users' representations may reveal private information, without access to the training data (and in fact, the same is true about the recommendations presented to the user -- i.e. the actual recommendation, which can be intercepted or seen by a third party, includes private information that we may wish to conceal). 

Privacy has also been studied recently in the context of recommender systems \cite{berlioz2015applying,nikolaenko2013privacy,mcsherry2009differentially,friedman2016differential,liu2015fast,shen2014privacy}. This growing body of work has been concerned for the most part with differentially private recommender systems, and always with the aim of guaranteeing that the actual records used to train the system are not recoverable. Unlike these works, we are concerned with the leakage of private information (such as demographics: age, gender, etc.) that was not directly present during training, but was implicitly learned by the system in the process of generating a useful user representation. These two ideas are not mutually exclusive, and may be combined to achieve a better privacy preserving recommender.

The problem of implicit private information studied here is closely related to the one studied in \cite{zemel2013learning}. In their work, they look for representations that achieve both {\it group fairness} and {\it individual fairness} in classification tasks. Individual fairness means that two persons having a similar representation should be treated similarly. Group fairness means that given a group we wish to protect (some proper subset $S \subset P$ of the population $P$), the proportion of people positively classified in $S$ equals that in the entire population. They achieved this goal by solving an optimization problem to learn an intermediate representation that obfuscates the membership in the protected group. In both cases the aim is to achieve good results on the respective predictive tasks, while using a representation that is agnostic to some aspects of the implicit structure of the domain. 

Several works apply adversarial training for the purpose of creating a representation free of a specific attribute \cite{beutel2017data,xie2017controllable,zhang2018mitigating,elazar:2018} (in fact, the original purpose of the method of adversarial training can also be seen as such). To the best of our knowledge, the current paper represents the first application of these ideas in the domain of recommender systems. Furthermore, while all the aforementioned applications experiment with solely one feature at a time, in this work we aim to create a representation free of multiple demographic features.

\section{\uppercase{Privacy and User Representations}}
User representations in recommendation systems are designed to encode the relevant information to determine user-item affinity. As such, to the extent that item affinity is dependent on certain user characteristics (such as demographics), the optimal user representations must include this information as well in order to have the necessary predictive power with respect to recommendations. We formalize this intuition using an information theoretic approach:

\begin{theorem}
\label{thm:1}
Let $\hat{v} = f(p_u, q_i)$ be an estimator of an outcome variable v associated with a pair $(p_u, q_i)$ of user and item representations (we use $p$ to denote users, and $q$ to denote items throughout the paper), and d any variable associated with users (such as age, gender, or marital status to name a few):  
\[ I(v;\hat{v}) > H[v|d] \implies I(\hat{v};d) > 0 \] 
\end{theorem}

\begin{proof}

\[ H[v] = H[v|d] + I(v;d) = H[v|\hat{v}] + I(v;\hat{v}) \]

\noindent rearranging and using $I(v;\hat{v}) > H[v|d]$ we have:

\[ I(v;d) > H[v|\hat{v}] \] 

\noindent and therefore:

\[ I(v;d) + I(v;\hat{v}) > H[v|\hat{v}] + H[v] - H[v|\hat{v}] = H[v]\]

\[ \implies I(\hat{v};d) > 0\]

\noindent To see how the final step follows, suppose on the contrary:

\[ I(\hat{v};d) = 0 \]

\noindent then from the above and by the chain rule for information we have:  

\[ H[v] < I(v;d) + I(v;\hat{v}) = I(v;d) + I(v;\hat{v}|d) = I(v;d, \hat{v}) \]

\noindent in contradiction to the $H[x] \geq I(x;\cdot)$ relation of entropy and mutual information.

\end{proof}

Theorem \ref{thm:1} asserts that the \textit{predictions} $\hat{v}$ must contain information about any relevant variable associated with users, if they are better than a certain threshold determined by the relevance of the variable. For example, if \textit{age} is a strong determinant of the movies a user is likely to want to watch, then by looking at the recommendations for a user we will be a able to extract some information about said user's age. Next, we show that the same is true about the \textit{user representations} used by the recommender, once the item representations are fixed (i.e. after learning): 

\begin{theorem}
\label{thm:2}
Let $\hat{v} = f(p_u, q_i) = g(p_u)$ be an estimator of outcome variable v, and d a variable associated with users, then: $I(p_u;d) \geq I(\hat{v};d) $
\end{theorem}

\begin{proof}
$\hat{v}$ is a function of $p_u$ alone and so by the data processing inequality we have $\forall d: I(p_u;d) \geq I(\hat{v};d)$. 
\end{proof}

\begin{corollary}
As a result of Theorems \ref{thm:1} - \ref{thm:2}, for any meaningful recommendation system and user characteristic there will be information leakage between the user representation and the characteristic. Specifically, for any system $\hat{v} = f(p_u, q_i)$, and user characteristic $d$ we have that:
\[ I(p_u;d) \geq I(\hat{v};d) > 0 \]
\end{corollary}

This assertion, that we cannot have both perfect privacy and performance in our setting, naturally leads to the question of the trade-off between the extent of information leakage, and the performance of the recommendation system. While the precise point selected on this trade-off curve is likely to be determined by the use-case, it would be reasonable to assume that in any case we will not want to sacrifice privacy unless we gain in performance. This can be understood as a Pareto optimality requirement on the multi-objective defined by the system and privacy objectives:

\begin{definition} 
The \textbf{privacy acceptable} subset of a family $\mathbb{S}$ of recommendation systems of the form: 
$f:[n_{users}] \times [n_{items}] \rightarrow \mathbb{R}_{+}$ with respect to a recommendation loss $l$ and a privacy target $h$ is the Pareto front in $\mathbb{S}$ of the multi-objective $(l, h)$.  
\end{definition}

While the method described in the rest of this paper does not directly address the issue of selecting a privacy-acceptable system, we show that our method is able to dramatically reduce information leakage while maintaining the majority of system performance. Future work will focus on methods and analytical tools in the spirit of Theorem \ref{thm:1} to assert that for a given system with a certain performance, there does not exist a system (in the family under consideration) with at least equal performance and better privacy. 

\subsection{Privacy-Adversarial Recommendation Systems}
In the previous section we showed that for any recommendation system where the user representations capture enough to make good predictions, these representations must reveal information about any pertinent user characteristic. In this section we describe a method to reduce such information from the user representations, in a way that allows to select along the trade-off curve of performance and information leakage. 

The method we use borrows the key idea from domain-adversarial training \cite{ganin2016domain}, where the aim is to learn a representation that is agnostic to the domain from which the example is drawn. Adapted to the problem at hand, this method enables us to construct a user representation from which the private information can not be read out. 

We start with an arbitrary latent factor recommendation system (which we will assume is trained using a gradient method). An additional construct is then appended from the user vectors, the output of which is the private field(s) we wish to censor. During training we follow two goals: (a) we would like to change the recommender parameters to optimize the original system objective, and (b) we would like to change the user vectors only, in order to \textbf{harm} the readout of the private information, while optimizing the readout parameters themselves with respect to the private information readout target. This is achieved by application of the gradient reversal trick introduced in \cite{ganin2016domain}, leading to the following update rule for user representation $p_u$:

\begin{equation}
p_u \leftarrow p_u -\alpha \Big[ \frac{\partial loss_{recsys}}{\partial p_u}  - \lambda \sum_i \frac{\partial loss_{demographics_i}}{\partial p_u} \Big]  
\end{equation}

\noindent where $\alpha$ and $\lambda$ are the general learning rate, and the adversarial training learning rate respectively. $loss_{recsys}$ is the recommendation system loss, and $loss_{demographics_i}$ is the loss for the $i-th$ demographic field prediction task. 

Two special cases are noteworthy. First, for $\lambda=0$ this formulation reduces back to the regular recommendation system. Second, setting $\lambda < 0$ we get the multi-task setting where we are trying to achieve both the recommendation and the demographic prediction task simultaneously. 

The gradient descent update for the rest of the recommendation system parameters (namely, the item representations and biases) is done by the regular $-\alpha\frac{\partial loss_{recsys}}{\partial \theta}$ update rule. Likewise, the parameters for the demographic field readouts are optimized in the same way with respect to their respective classification objectives. 

\section{\uppercase{A note about fairness}}
The issue of fairness of models is a long standing debate in the scientific community and beyond. Some cases aspects of fairness are mandated by legislation or social norms. For instance, when modeling risk for the purpose of loans, in most countries the use of certain characteristics (such as gender or race) would be strictly prohibited. However, while these protected variables are not explicitly entered into the model, how can we be sure that they did not affect the outcome via correlations with variables which were indeed included? 

Many methods have been devised to address this question and guarantee fairness (see Section \ref{sec:related} for a brief summary of some of these lines of work). The ultimate solution to the problem of fairness is to be able to  guarantee that a set of variables did not have an affect on the outcome of a model. In order to be able to give such a guarantee we should be able to assert that the \textbf{information content} of the forbidden variables was not present in the model. 

It is easy to see why just excluding protected variables from the model is not enough. Suppose for instance we wish to exclude \textit{gender} from our model in order to have equality in the outcome of the model with respect to this attribute. Namely, we want the distributions of outcomes to be the same for all genders. The first step is to exclude this variable from the model. However, including other seemingly benign attributes such as \textit{occupation} will allow \textit{gender} to be implicitly included in the model (since presumably these variables are dependent, i.e. share information that will allow the model to 'guess' the gender anyway). 

The method of privacy-adversarial training of recommender systems presented in this paper is a step towards fairness in machine learning models. By applying privacy-adversarial training with respect to a set of protected attributes, we ensure that the representation of the individual does not include these attributes, and they therefore can't affect the outcome. However, we note that the method does not come with a provable guarantee that this information was indeed scrubbed from the model. Furthermore, in a recent study \cite{elazar:2018} it has been shown that a similar technique is able to eliminate demographic classification during training in an NLP task, but the text representations produced by the model still contain demographic information that can be extracted via different classifiers. In the next section we provide evidence that at least in the context of recommendation systems the method of privacy-adversarial training works well, and that the resulting representations do not contain information about gender and age that can be read out even by additional classifier (Table \ref{tab:all-clf}). That being said, it is important to remember that for critical issues such as privacy and fairness we would ideally want provability of the properties of the method.

\section{\uppercase{Results}}
\paragraph{Data} The experiments in this section were conducted on the MovieLens 1M dataset \cite{harper2016movielens}. This extensively studied dataset (see for example \cite{miller2003movielens,chen2010social,jung2012attribute,peralta2007extraction}) includes 1,000,209 ratings from 6,040 users, on 3,706 movies. In addition, demographic information in the form of age and gender is provided for each user. Gender (male/female) is skewed towards male with 71.7\% in the male category. Age is divided into $7$ groups (0-18, 18-25, 25-35, 35-45, 45-50, 50-56, 56-inf) with 34.7\% in the most popular age group, being 25-35.
This means that when absolutely no personal data is given about an arbitrary user, the prediction accuracy of gender and age group cannot exceed 71.7\% and 34.7\% respectively.

\paragraph{Recommendation System} We use the Bayesian Personalized Ranking (BPR) recommendation system \cite{rendle2009bpr}, a natural choice for ranking tasks. The model is modified with adversarial demographic prediction by appending a linear readout structure from the user-vectors to each of the demographic targets (binary gender and 7-category age). The \textit{gradient reversal layer (GRL)} \cite{ganin2016domain}  is applied between the user-vectors and each of the demographic readout structures, so that effectively during training the user vectors are optimized with respect to the recommendation task, but de-optimized with respect to the gender prediction. At the same time, the demographic readout is optimized to use the current user representation to predict gender and age. The result of this scheme is a user representation that is good for recommendation but not good for demographic prediction (i.e. is purged of the information we do not want it to contain). We note that the same method could be applied to any type of recommendation system which includes user representations and is trained using a gradient based method.

\paragraph{Evaluation} Recommendation systems were evaluated using a hold out set. For each user in the MovieLens 1M Dataset, the final movie that they watched was set aside for testing, and never seen during training. The fraction of users for whom this held out movie was in the top-k recommendation \cite{koren2008factorization} is reported as model accuracy (we use $k=10$). 
Private information in the user representations was evaluated using both a neural-net predictor of the same form used during adversarial training, and a host of standard classifiers (SVMs with various parameters and kernels, Decision Trees, Random Forest -- see Table \ref{tab:all-clf}). The rest of the results are shown for the original neural classifier with a cross validation procedure. Results are reported as accuracy. 

\begin{table}
  \caption{Verification of the inability to predict demographic fields from user representations trained in the privacy-adversarial method. Results in this table are given for representations of size $10$ with $\lambda=1.$}
  \label{tab:all-clf}
  \begin{tabular}{lcc}
    \toprule
    classifier & gender & age \\
    \midrule
    large class baseline & 71.70 & 34.70 \\
    \midrule
    softmax neural net readout & 71.70 & 34.47 \\ 
    SVM (linear; C=.1)   & 71.71 & 34.62 \\ 
    SVM (linear; C=1)    & 71.71 & 34.59 \\ 
    SVM (linear; C=10)   & 71.71 & 34.59 \\ 
    SVM (RBF kernel)     & 71.71 & 29.20 \\ 
    Decision Tree        & 64.00 & 25.44 \\ 
    Random Forrest       & 69.80 & 29.20 \\ 
    Gradient Boosting    & 71.79 & 33.48 \\
  \bottomrule
\end{tabular}
\end{table}

\begin{table}
  \caption{Gender prediction from user representations. First column corresponds to the regular recommendation system, and the following columns to privacy-adversarial training with the prescribed value of $\lambda$. Rows correspond to the size of user and item representations. Final row contains the na\"ive baseline reverting to the predicting the largest class.}
  \label{tab:gender}
  \begin{tabular}{c|c|cccc}
    \toprule
    size / $\lambda$ & 0 & .01 & .1 & 1 & 10 \\
    \midrule
    10 & 76.97 & 74.21 & 71.33  & \textbf{71.70} & 71.60 \\
    20 & 77.55 & 74.26 & 71.50  & 72.30 & 71.03 \\
    50 & 77.80 & 74.34 & 72.24  & 86.00 & 74.26 \\ 
    \midrule
    na\"ive &  \multicolumn{2}{c}{$\cdots$} & 71.70\% & \multicolumn{2}{c}{$\cdots$} \\ 
  \bottomrule
\end{tabular}
\end{table}

\paragraph{Results-privacy} Private demographic information does indeed exist in user representations in the standard recommendation system. Gender prediction (Table \ref{tab:gender}, $\lambda=0$ column) increases with size of user representation to 77.8\% (recall 71.7\% are Male). Likewise, age bracket prediction also increases with size of user representation and reaches 44.90\% (largest category is 34.7\%). These results serve as the baseline against which the adversarial training models are tested against. Our aim in the privacy-adversarial setting will be to reduce the classification results down to the baseline, reflecting user representations were purged of this private information.

\begin{table}[t]
  \caption{Age prediction from user representations. First column corresponds to the regular recommendation system, and the following columns to privacy-adversarial training with the prescribed value of $\lambda$. Rows correspond to the size of user and item representations. Final row contains the na\"ive baseline reverting to the predicting the largest class.}
  \label{tab:age}
  \begin{tabular}{c|c|cccc}
    \toprule
    size / $\lambda$ & 0 & .01 & .1 & 1 & 10 \\
    \midrule
    10 & 41.29 & 36.28 & 34.29 & 34.47 & 34.31 \\
    20 & 44.16 & 36.34 & 33.97 & 35.70 & 38.01 \\
    50 & 44.90 & 36.46 & \textbf{34.62} & 70.27 & 52.7 \\
	\midrule 
    na\"ive &  \multicolumn{2}{c}{$\cdots$} & 34.70\% & \multicolumn{2}{c}{$\cdots$} \\ 

  \bottomrule
\end{tabular}
\end{table}

\paragraph{Results-privacy-adversarial} In the privacy-adversarial setting, overall prediction results for both gender and age are diminished to the desired level of the largest class. With $\lambda=.1$, for example, age prediction is eliminated completely (reducing effectively to the 34.7\% baseline) for all sizes of representation, and likewise for gender with representation of size $10-20$. For size $50$ we see some residual predictive power, though it is highly reduced relative to the regular recommendation system. 

For the large representation (size $50$) and large values of $\lambda$ in the range of $\lambda \geq 1$ we see an interesting phenomenon of reversal of the effect, with demographic readout sometimes way above the regular recommendation system (e.g when the embedding size = 50 and $\lambda=1$ the gender prediction achieves 86.0\%). We suspect this happens due to the relative high learning rate, which causes the system to diverge.

\begin{table}[b]
  \caption{Recommendation System performance (accuracy@10) with privacy-adversarial training. First column corresponds to the regular recommendation system, and the following columns to privacy-adversarial training with the prescribed value of $\lambda$. Rows correspond to the size of user and item representations.}
  \label{tab:recsys}
  \begin{tabular}{c|c|cccc}
    \toprule
    size / $\lambda$ & 0 & .01 & .1 & 1 & 10 \\
    \midrule
    10 & \textbf{3.05} & 2.76 & \textbf{2.88} & 2.43 & 2.67 \\
    20 & 3.00 & 2.68 & 2.04 & 2.20 & 2.07 \\
    50 & 2.65 & 2.38 & 2.22 & 2.22 & 2.15 \\

  \bottomrule
\end{tabular}
\end{table}

With respect to the trade-off between system performance and privacy, results indicate (Table. \ref{tab:recsys}) that smaller user representations (size $10$) are preferential for this small dataset. We see some degradation with adversarial training, but nevertheless we are able to eliminate private information almost entirely with representations of size $10$ and $\lambda=.1$ while sacrificing only a small proportion of performance (accuracy@10 of 2.88\% instead of the 3.05\% for the regular system, gender information gap of 0.37\% and age information gap of 0.41\% from the majority group). 

Together, these results show the existence of the privacy leakage in a typical recommendation system, and the ability to eliminate it with privacy-adversarial training while harming the overall system performance only marginally. 

\section{\uppercase{Conclusions}}

In this paper we discuss information leakage from user representations of Latent factor recommender systems, and show that private demographic information can be read-out even when not used in the training data. We adapt the adversarial training framework in the context of privacy in recommender systems. An adversarial component is appended to the model for each of the demographic variables we want to obfuscate, so that the learned user representations are optimized in a way that precludes predicting these variables. We show that the proposed framework has the desired privacy preserving effect, while having a minimal overall adverse effect on recommender performance, when using the correct value of the trade-off parameter $\lambda$. Our experiments show that this value should be determined for a given dataset, since values too large lead to instability of the adversarial component. But in any case, as suggested by \cite{elazar:2018}, when concerened with sensitive features, one should verify the amount of information in the representation with additional post-hoc classifier.

The adversarial method can be used to obfuscate any private variable known during training (in this paper we discuss categorical variables, but the generalization to the continuous case is trivial). While at first glance this may be seen as a shortcoming of the approach, it is interesting to note that it would be inherently infeasible to force the representation not to include \textit{any} factor implicitly associated with item choice. Clearly, in such a case there would be no information left to drive recommendations. The intended use of the method is rather to hide a small set of protected variables known during training, while using the rest of the implicit information in the usage data to drive recommendations.       

An interesting topic for further research is the amount of private information that is available in the top-k recommendations themselves. Since the sole reason private demographic information is present in the user representations is to help drive recommendations, it stands to reason that it would be possible to design a method of reverse-engineering in the form of a readout from the actual recommendations. Such a leakage, to the extent that it indeed exists, would have much further reaching practical implications for privacy and security of individuals.  


Another topic for further research is the use of privacy-adversarial training to boost the personalization and specificity of recommendations. By eliminating the demographic (or other profile related) information, suggested items are coerced out of stereotypical templates related to coarse profiling. It is our hope that user testing will confirm that this leads to deeper and more meaningful user models, and overall higher quality recommendations. 

\bibliographystyle{apalike}
{\small
\bibliography{lib}}


\vfill
\end{document}